\documentclass{article}

\usepackage{graphicx}
\usepackage[space]{grffile}
\usepackage{latexsym}
\usepackage{textcomp}
\usepackage{longtable}
\usepackage{multirow,booktabs}
\usepackage{amsfonts,amsmath,amssymb}

\usepackage{longtable}

\usepackage{multirow,booktabs}

\usepackage{amsthm,amsfonts,amsmath,amssymb}

\usepackage{url}

\usepackage{hyperref}

\hypersetup{colorlinks=false,pdfborder={0 0 0}}
\newif\iflatexml\latexmlfalse
\usepackage[utf8]{inputenc}

\usepackage[normalem]{ulem}

\usepackage{tikz}

\usepackage{comment}
\newtheorem{definition}{Definition}
\newtheorem{lemma}{Lemma}
\newtheorem{example}{Example}
\newtheorem{theorem}{Theorem}
\newtheorem{fact}{Fact}

\newtheorem{corollary}{Corollary}

\renewcommand{\L}{\mathcal{L}}
\newcommand{\N}{N}

\newcommand{\I}{\Atoms}
\newcommand{\Model}{\mathcal{M}}
\renewcommand{\S}{\mathcal{A}}
\newcommand{\A}{{\bf A}}
\newcommand{\G}{{\bf G}}
\newcommand{\Atoms}{{\bf P}}

\renewcommand{\O}{{\bf O}}

\newcommand{\D}{\mathcal O}
\newcommand{\AB}{\mathcal{O}^\ast}
\newcommand{\maj}{\mathsf{maj}}
\newcommand{\pv}{\mathsf{pv}}

\newcommand{\tuple}[1]{\left\langle #1 \right\rangle}
\newcommand{\set}[1]{\left\{ #1 \right\}}

\newcommand{\0}{{\bf 0}}
\newcommand{\1}{{\bf 1}}

\renewcommand{\phi}{\varphi}

\renewcommand{\phi}{\varphi}

\newcommand{\true}[1]{\lVert #1 \rVert}

\newcommand{\K}{\mathsf{K}}

\newcommand{\lequiv}{\leftrightarrow}
\newcommand{\limp}{\rightarrow}
\newcommand{\ldia}[1]{\left\langle #1 \right\rangle}
\newcommand{\lbox}[1]{\left[ #1 \right]}

\newcommand{\Stb}{\mathsf{stb}}

\newcommand{\AND}{\mbox{   \textsc{and}   }}
\newcommand{\IFF}{\Longleftrightarrow}

\makeindex

\begin{document}

\title{Liquid Democracy: \\ An Analysis in Binary Aggregation and Diffusion\footnote{{\bf Working paper}: The paper collects work presented at: {\em Dynamics in Logic IV}, TU Delft, November 2016; seminars at the Computer Science Departments of the University of Leicester and the University of Oxford, December 2016; {\em Dutch Social Choice Colloquium}, December 2016. The authors wish to thank the participants of the above workshops and seminars for many helpful suggestions. The authors wish also to thank Umberto Grandi for many insightful comments on an earlier version of this paper. Both authors acknowledge support for this research by EPSRC under grant EP/M015815/1.}}

\author{Zo\'e Christoff$^{\star}$ and Davide Grossi$^{\star\star}$ \\
$^{\star}$
Department of Philosophy,
University of Bayreuth \\
\url{zoe.christoff@gmail.com}\\
$^{\star\star}$
Department of Computer Science, University of Liverpool \\
\url{d.grossi@liverpool.ac.uk}
}



\maketitle

\begin{abstract}
The paper proposes an analysis of liquid democracy (or, delegable proxy voting) from the perspective of binary aggregation and of binary diffusion models. We show how liquid democracy on binary issues can be embedded into the framework of {\em binary aggregation with abstentions}, enabling the transfer of known results about the latter---such as impossibility theorems---to the former. This embedding also sheds light on the relation between delegation cycles in liquid democracy and the probability of collective abstentions, as well as the issue of individual rationality in a delegable proxy voting setting. We then show how liquid democracy on binary issues can be modeled and analyzed also as a specific process of dynamics of binary opinions on networks. These processes---called {\em Boolean DeGroot processes}---are a special case of the DeGroot stochastic model of opinion diffusion. We establish the convergence conditions of such processes and show they provide some novel insights on how the effects of delegation cycles and individual rationality could be mitigated within liquid democracy.

The study is a first attempt to provide theoretical foundations to the delgable proxy features of the liquid democracy voting system. Our analysis suggests recommendations on how the system may be modified to make it more resilient with respect to the handling of delegation cycles and of inconsistent majorities.

\end{abstract}

\newpage

{\small
\tableofcontents
}
\newpage


\section{Introduction}

Liquid  democracy \cite{liquid_feedback} is a form of democratic decision-making considered to stand between direct and representative democracy. It has been used, advocated and popularized by local and even national parties (e.g., Demoex\footnote{\url{demoex.se/en/}} in Sweden, and Piratenpartei\footnote{\url{www.piratenpartei.de}} in Germany) to coordinate the behavior of party representatives in assemblies, as well as campaigns (e.g., Make Your Laws\footnote{\url{www.makeyourlaws.org}} in the US). At its heart is voting via a delegable proxy, also called sometimes transitive proxy. For each issue submitted to vote, each agent can either cast its own vote, or it can delegate its vote to another agent---a proxy---and that agent can delegate in turn to yet another agent and so on. This differentiates liquid democracy from standard proxy voting \cite{Miller_1969,Tullock_1992}, where proxies cannot delegate their vote further. Finally, the agents that decided not to delegate their votes cast their ballots (e.g., under majority rule, or adaptations thereof), but their votes now carry a weight consisting of the number of all agents that, directly or indirectly, entrusted them with their vote.

\paragraph{Scientific context and contribution}
Analyses of standard (non-delegable) proxy voting from a social choice-theoretic perspective---specifically through the theory of spatial voting---have been put forth in \cite{Alger_2006} and \cite{Green_Armytage_2014}. Delegable proxy has not, to the best of our knowledge, been object of study so far, with the notable exception of \cite{Boldi_2011} which focuses specifically on algorithmic aspects of a variant of liquid democracy (which the authors refer to as {\em viscous democracy}) with applications to recommender systems.


The objective of the paper is to provide a first analysis, via formal methods, of the liquid democracy voting system based on delegable proxy. This, we hope, should point to a number of future lines of research and stimulate further investigations into this and related systems.

\paragraph{Outline}
The paper starts in Section \ref{sec:preliminaries} by introducing some preliminaries on the theory of binary aggregation, which is the framework of reference for this study. It is then structured in two parts. This preliminary section presents also novel results on binary aggregation with abstentions. The first part (Section \ref{sec:proxy}) studies voting in liquid democracy from the point of view of the delegation of voting power: we study delegable proxy aggregators using the machinery of binary and judgment aggregation. This allows us to shed novel light on some issues involved in the liquid democracy system, in particular: the issue of circular delegation, and the issue of individual irrationality when voting on logically interdependent issues. The second part (Sections \ref{sec:diffusion} and \ref{sec:logic}) studies voting in liquid democracy as a very specific type of opinion diffusion on networks, whereby delegation is rather interpreted as the willingness to copy the vote of a trustee. We show that this perspective provides some interesting insights on how to address the above mentioned issues of circular delegations and individual irrationality. Section \ref{sec:conclusions} concludes the paper and outlines some on-going lines of research.



\section{Binary Aggregation with Abstention} \label{sec:preliminaries}

The formalism of choice for this paper is binary aggregation \cite{grandi13lifting} with abstention.\footnote{The standard framework of binary aggregation without abstention is sketched in the appendix for ease of reference.} This preliminary section is devoted to its introduction.

\subsection{Opinions and Opinion Profiles}

A binary aggregation structure (\emph{BA structure}) is a tuple $\S = \tuple{\N,\Atoms,\gamma}$ where:
\begin{itemize}
\item $\N = \set{1,\dots,n}$ is a non-empty finite set individuals s.t. $|\N|= n \in \mathbb{N}$;
\item $\Atoms = \set{p_1,\dots,p_m}$ is a non-empty finite set of issues ($|\Atoms|= m \in \mathbb{N}$), each represented by a propositional atom;
\item $\gamma \in \L$ is an (integrity) constraint, where $\L$ is the propositional language constructed by closing $\Atoms$ under a functionally complete set of Boolean connectives (e.g., $\set{\neg, \wedge}$).
\end{itemize}

An {\em opinion} function $O$ is an assignment acceptance/rejection values (or, truth values) to the set of issues $\Atoms$. Thus, $O(p)=\0$ (respectively, \mbox{$O(p)=\1$}) indicates that opinion $O$ rejects (respectively, accepts) the issue $p$. Syntactically, the two opinions correspond to the truth of the literals $p$ or $\neg p$. For $p \in \Atoms$ we write $\pm p$ to denote one element from $\set{p, \neg p}$, and $\pm \Atoms$ to denote $\bigcup_{p\in\Atoms} \set{p, \neg p}$, which we will refer to as the {\em agenda} of $\S$. Allowing abstention in the framework of binary aggregation amounts to considering incomplete opinions: an {\em incomplete opinion} is a partial function from $\Atoms$ to $\set{\0,\1}$. We will study it as a function $O: \Atoms \rightarrow \set{\0,\1, \ast}$ thereby explicitly denoting the undetermined value corresponding to abstention. 

We say that the incomplete opinion of an agent $i$ is \emph{consistent} if the set of formulas $\set{p \mid O_i(p) = \1} \cup \set{\neg p \mid O_i(p) = \0} \cup \set{\gamma}$ can be extended to a model of $\gamma$ (in other words, if the set is satisfiable). Intuitively, the consistency of an incomplete opinion means that the integrity constraint is consistent with $i$'s opinion on the issues she does not abstain about. We also say that an incomplete opinion is {\em closed} whenever the following is the case: {\em if} the set of propositional formulas $\set{p \mid O_i(p) = \1} \cup \set{\neg p \mid O_i(p) = \0} \cup \set{\gamma}$ logically implies $p$ (respectively, $\neg p$), {\em then} $O_i(p) = 1$ (respectively, $O_i(p) = 0$). That is, individual opinions are closed under logical consequence or, in other words, agents cannot abstain on issues whose acceptance or rejection is dictated by their expressed opinions on other issues. The set of incomplete opinions is denoted $\AB$ and the set of consistent and closed incomplete opinions $\AB_c$. As the latter are the opinions we are interested in, we will often refer to them simply as individual opinions.

An \emph{opinion profile} $\O = (O_1,\dots,O_{n})$ records the opinion, on the given set of issues, of every individual in $\N$. Given a profile $\O$ the $i^{\mathit{th}}$ projection $\O$ is denoted $O_i$ (i.e., the opinion of agent $i$ in profile $\O$). Let us introduce some more notation.
We also denote by $\O(p)= \set{i \in \N \mid O_{i}(p)= \1}$ the set of agents accepting issue $p$ in profile $\O$, by $\O(\neg p)= \set{i \in \N \mid O_{i}(p)= \0}$ and by $\O(\pm p) = \O(p) \cup \O(\neg p)$ the set of non-abstaining agents. We write $\O =_{-i} \O'$ to denote that the two profiles $\O$ and $\O'$ are identical, except for possibly the opinion of voter $i$.

\subsection{Aggregators}

Given a BA structure $\S$, an \emph{aggregator} (for $\S$) is a function $F:(\AB_{c})^\N \to \AB$, mapping every profile of individual opinions to one collective (possibly incomplete) opinion.\footnote{It is therefore worth stressing that, in this paper, we study aggregators that are resolute (that is, output exactly one value), even though they allow for collective abstention.} $F(\O)(p)$ denotes the outcome of the aggregation on issue $p$. The benchmark aggregator is the \emph{issue-by-issue strict majority rule} ($\maj$), which accepts an issue if and only if the majority of the non-abstaining voters accept that issue: 
\begin{align}\label{eq:majast}
\maj(\O)(p)=  
\begin{cases}
        \1  & \mathit{ if }  |\O(p)|  > |\O(\neg p)|\\
        \0    & \mathit{ if }  |\O(\neg p)| > |\O(p)| \\
       \ast   & \mathit{ otherwise }   \\
\end{cases}
\end{align}
We will refer to this rule simply as `majority'.

Majority can be thought of as a 
quota rule. In general, quota rules in binary aggregation with abstention are of the form: accept when the proportion of \emph{non-abstaining} individuals accepting is above the acceptance-quota, reject when the proportion of \emph{non-abstaining} individuals is above the rejection-quota, and abstain otherwise:\footnote{There are several ways to think of quota rules in the presence of abstentions. Instead of a quota being a proportion of non-abstaining agents, one could for instance define rules with absolute quotas instead: accept when at least $n$ agents accept, independently of how many agents do not abstain. In practice, voting rules with abstention are often a combination of those two ideas: accept an issue if a big enough proportion of the population does not abstain, and if a big enough proportion of those accept it.}

\begin{definition}[Quota rules] \label{def:quota}
Let $\S$ be an aggregation structure.  
A {\em quota rule} (for $\S$) is defined as follows, for any issue $p\in\Atoms$, and any opinion profile $\O\in\AB$:\footnote{The definition uses the ceiling function $\lceil x \rceil$ denoting the smallest integer larger than $x$.}
\begin{align}\label{eq:quotarules}
F(\O)(p)=  
\begin{cases}
        \1  & \mbox{ if }  |\O(p)|  \geq \left\lceil q_\1(p) \cdot |\O(\pm p)| \right\rceil\\
        \0  & \mbox{ if }  |\O(\neg p)| \geq \left\lceil q_\0(p) \cdot |\O(\pm p)| \right\rceil  \\
       \ast & \mbox{ otherwise } \\
\end{cases}
\end{align} \label{eq:quota}
where for $x \in \set{\0, \1}$, $q_x$ is a function $q_x: \Atoms \to (0,1] \subset \mathbb{Q}$ assigning a positive rational number smaller or equal to $1$ to each issue, and such that, for each $p \in \Atoms$:
\begin{align}
q_x(p) > 1 - q_{(\1 - x)}(p), \label{eq:constraint}
\end{align}
A quota rule is called: {\em uniform} if, for all $p_i,p_j \in \Atoms$, $q_x(p_i) = q_x(p_j)$;
it is called {\em symmetric} if, for all $p \in \Atoms$, $q_\1(p) = q_\0(p)$.
\end{definition}
Notice that the definition excludes trivial quota.\footnote{Those are quotas with value $0$ (always met) or $>1$ (never met). Restricting to non-trivial quota is not essential but simplifies our exposition.} 
It should also be clear that, by \eqref{eq:constraint} the above defines an aggregator of type $(\AB_{c})^\N \to \AB$ as desired.\footnote{What needs to be avoided here is that both the acceptance and rejection quota are set so low as to make the rule output both the acceptance and the rejection of a given issue} Notice also that if the rule is symmetric, then \eqref{eq:constraint} forces $q_x > \frac{1}{2}$.

\begin{example} \label{example:maj}
The majority rule \eqref{eq:majast} is a uniform and symmetric quota rule where $q_\1$ and $q_\0$ are set to meet the equation $\lceil q_\1(p) \cdot |\O(\pm p)| \rceil = \lceil q_\0(p) \cdot |\O(\pm p)| \rceil = \left\lceil \frac{|\O(\pm p)| + 1}{2}\right\rceil$, for any issue $p$ and profile $\O$. This is achieved by setting the quota as $\frac{1}{2}<q_\1,q_\0 \leq \frac{1}{2}+\frac{1}{|N|} = \frac{|N|+1}{2|N|}$. More precisely one should therefore consider $\maj$ as a class of quota rules yielding the same collective opinions.
\end{example}

\begin{example}
The uniform and symmetric unanimity rule is defined by setting $q_\1 = q_\0 = 1$. A natural uniform but asymmetric variant of unanimity can be obtained by setting $q_\1 = 1$ and $q_\0 = \frac{1}{|\N|}$.
\end{example}

Let us finally note an important difference between quota rules in binary aggregation with abstentions vs. without abstentions. In a framework without abstentions quota rules are normally defined by a unique acceptance quota $q^\1$, the rejection quota being uniquely determined as $q^\0 = 1 - q^\1$. As a consequence, the majority rule, when $|N|$ is odd, is the only unbiased quota rule in the standard framework. This is no longer the case when abstentions are considered. A novel characterization of the majority rule will be given in Section \ref{subsec:char.quota}.


\subsection{Agenda conditions}


\begin{definition}[simple/evenly negatable agenda]
An agenda $\pm \I$ is said to be {\em simple} if there exists no set $X \subseteq \pm \Atoms$  such that: $|X|\geq 3$, and  $X$  is minimally $\gamma$-inconsistent, that is:
\begin{itemize}
\item $X$ is inconsistent with $\gamma$
\item For all  $Y\subset X$, $Y$ is consistent with $\gamma$ (or, $\gamma$-consistent).
\end{itemize}
An agenda is said to be {\em evenly negatable} if there exists a minimal $\gamma$-inconsistent set $X \subseteq \pm \Atoms$ such that for a set $Y \subseteq X$ of even size, $X\backslash Y \cup \set{\neg p \mid p \in Y}$ is $\gamma$-consistent. It is said to be {\em path-connected} if there exists $p_1, \ldots, p_n \in \pm \I$ such that $p_1 \models^c p_2, \ldots, p_{n-1} \models^x p_n$ where $p_i \models^c p_{i+1}$ (conditional entailment) denotes that there exists $X \subseteq \pm\I$, which is $\gamma$-consistent with both $p_i$ and $\neg p_{i+1}$, and such that $\set{p} \cup X \cup \set{\gamma}$ logically implies $p_{i +1}$. 
\end{definition}
We refer the reader to \cite[Ch. 2]{Grossi_2014} for a detailed exposition of the above conditions. We provide just a simple illustrative example.

\begin{example}
Let $\I = \set{p, q, r}$ and let $\gamma = (p \wedge q ) \rightarrow r$. $\pm \I$ is not simple. The set $\set{p, q, \lnot r} \subseteq \pm\Atoms$ is inconsistent with $\gamma$, but none of its subsets is. 
\end{example}


\subsection{Properties of aggregators}\label{subsec:axiomatic}

We start by recalling some well-known properties of aggregators from the judgment aggregation literature, adapted to the setting with abstention: 

\begin{definition} \label{def:properties}
Let $\S$ be an aggregation structure. An aggregator $F: (\AB_{c})^\N \to \AB$ is said to be:
\begin{description}
\item[unanimous] iff for all $p\in \Atoms$, for all profiles $\O$ and all $x\in\{0,1,\ast\}$: if for all $i\in \N, O_i(p) = x$, then $F(\O)(p)= x $. I.e., if everybody agrees on a value, that value is the collective value.
\item[anonymous] iff for any bijection $\mu: \N\rightarrow\N$, $F(\O)=F(\O^\mu)$, where $\O^\mu = \tuple{O_{\mu(1)}, \ldots, O_{\mu(n)}}$. I.e., permuting opinions among individuals does not affect the output of the aggregator.
\item[$p$-dictatorial] iff there exists $i \in \N$ (the {\em $p$-dictator}) s.t. for any profile $\O$, and all $x \in \set{\0,\1}$, $O_i(p) = x$ iff $F(\O)(p) = x$.
I.e., there exists an agent whose definite opinion determines the group's definite opinion on $p$. If $F$ is $p$-dictatorial, with the same dictator on all issues $p \in \Atoms$, then it is called {\bf dictatorial}.
\item[$p$-oligarchic] iff there exists $C \subseteq \N$ (the {\em $p$-oligarchs}) s.t. $C\neq\emptyset$ and for any profile $\O$, and any value $x \in \set{\0,\1}$, $F(\O)(p) = x$ iff $O_i(p) = x$ for all $i\in C$.
I.e., there exists a group of agents whose definite opinions always determine the group's definite opinion on $p$. If $F$ is $p$-oligarchic, with the same oligarchs on all issues $p \in \Atoms$, then it is called {\bf oligarchic}.
\item[monotonic] iff, for all $p\in \Atoms$ and all $i\in\N$: 
for any profiles $\O, \O'$, if $\O =_{-i} \O'$: (i) if $O_i(p)\neq \1$ and $O'_i(p)\in\{\1,\ast\}$, then: if $F(\O)(p)=\1$, then $F(\O')(p)=\1$; and (ii) if $O_i(p)\neq \0$ and $O'_i(p)\in\{\0,\ast\}$, then: if $F(\O)(p)=\0$, then $F(\O')(p)=\0$. I.e., increasing support for a definite collective opinion does not change that collective opinion.
\item[independent] iff, for all $p\in \Atoms$, for any profiles $\O, \O'$: if for all $i\in \N, O_i(p) = O'_i(p)$, then $F(\O)(p)=F(\O')(p)$. I.e., the collective opinion on each issue is determined only by the individual opinions on that issue.
\item[neutral] iff, for all $p,q \in \Atoms$, for any profile $\O$: if for all $i\in\N$, $O_i(p)=O_i(q)$, then $F(\O)(p)=F(\O)(q)$. I.e., all issues are aggregated in the same manner.
\item[systematic] iff it is neutral and independent. I.e., the collective opinion on issue $p$ depends only on the individual opinions on this issue.
\item[responsive] iff for all $p\in\Atoms$, there exist profiles $\O, \O'$  such that $F(\O)(p)=\1$ and $F(\O')(p)=\0$. I.e., the rule allows for an issue to be accepted for some profile, and rejected for some other. 
\item[unbiased] iff for all $p \in \Atoms$, for any profiles $\O, \O'$ : if for all $i\in\N$, $O_i(p)= \1$ iff $O'_i(p)=\0$ (we say that $\O'$ is the ``reversed'' profile of $\O$), then $F(\O)(p)=\1$ iff $F(\O')(p)=\0$. I.e., reversing all and only individual opinions on an issue $p$ (from acceptance to rejection and from rejection to acceptance) results in reversing the collective opinion on $p$. 
\item[rational] iff for any profile $\O$, $F(\O)$ is consistent and closed. I.e., the aggregator preserves the constraints on individual opinions.
\end{description}
\end{definition}

\begin{example}
It is well-known that majority is not rational in general. The standard example is provided by the so-called discursive dilemma, represented by the BA structure $\tuple{\set{1,2,3},\set{p,q,r},r \lequiv (p \land q)}$. The profile consisting of $O_1 \models p \land q \land r$, $O_2 \models p \land \neg q \land \neg r$, $O_3 \models \neg p \land q \land \neg r$, returns an inconsistent majority opinion $\maj(\O) \models p \land q \land \neg r$ (cf. \cite{Grossi_2014}).
\end{example}

Finally, let us defined also the following property. The {\bf undecisiveness} of an aggregator $F$ on issue $p$ for a given aggregation structure is defined as the number of profiles which result in collective abstention on $p$:
\begin{align}
u(F)(p) & = |\set{\O \in \AB_c \mid F(\O)(p) = \ast}|.
\end{align}


\subsection{Characterizing quota rules}\label{subsec:char.quota}

As a typical example, consider the aggregator $\maj$: it is unanimous, anonymous, monotonic, systematic, responsive and unbiased, but, as mentioned above, it is not rational in general.
However, it can be shown (cf. \cite[3.1.1]{Grossi_2014}) that aggregation by majority is collectively rational under specific assumptions on the constraint:

\begin{fact} \label{fact:rat}
Let $\S$ be a BA structure with a simple agenda. Then $\maj$ is rational.
\end{fact}
\begin{proof}
If the agenda $\pm\I$ is simple, then all minimally inconsistent sets have cardinality $2$, that is, are of the form $\set{\phi,\neg \phi}$ such that $\phi \models \neg \phi$ for $\phi,\psi \in \I$. W.l.o.g. assume $\phi = p_i$ and $\psi = p_j$. Suppose towards a contradiction that there exists a profile $\O$ such that $\maj(\O)$ is inconsistent, that is, $\maj(\O)(p_i) = \maj(\O)(p_j) = 1$, and $\phi \models \neg \psi$. By the definition of $\maj$ \eqref{eq:majast} it follows that $|\O(p_i)| > |\O(\neg p_i)|$ and $|\O(p_j)| > |\O(\neg p_j)|$. Since $p_i \models \neg p_j$ by assumption, and since individual opinions are consistent and closed, $|\O(\neg p_j)| \geq |\O(p_i)|$ and $|\O(\neg p_i)| \geq |\O(p_j)|$. From the fact that $|\O(p_i)| > |\O(\neg p_i)|$ we can thus conclude that $|\O(\neg p_j)|> |\O(p_j)|$. Contradiction.
\end{proof}

May's theorem \cite{May_1952} famously shows that for preference aggregation, the majority rule is in fact the \emph{only} aggregator satisfying a specific bundle of desirable properties. A corresponding characterization of the majority rule is given in judgment aggregation \emph{without abstention}: when the agenda is simple, the majority rule is the only aggregator which is rational, anonymous, monotonic and unbiased \cite[Th. 3.2]{Grossi_2014}. We give below a novel characterization theorem, which takes into account the possibility of abstention (both at the individual and at the collective level). To the best of our knowledge, this is the first result of this kind in the literature on judgment and binary aggregation with abstention.  


\smallskip

We first prove the following lemma:
\begin{lemma} \label{lemma:min}
Let $F$ be a uniform and symmetric quota rule for a given $\S$. The following holds:
$\frac{1}{2}< q_\1 = q_0 \leq \frac{|N|+1}{2|N|}$ if and only if $F = \arg\min_G u(G)(p)$, for all $p \in \Atoms$. 
\end{lemma}
\begin{proof}
The claim is proven by the following series of equivalent statements. 
(a) A uniform and symmetric quota rule $F$ has quota such that $\frac{1}{2}< q_\1 = q_0 \leq \frac{|N|+1}{2|N|}$. 
(b) A uniform and symmetric quota rule $F$ has quota such that $\lceil q_\1(p)|\O(\pm p)| \rceil = \lceil q_\0(p) |\O(\pm p)| \rceil = \left\lceil \frac{|\O(\pm p)| + 1}{2}\right\rceil$ for any profile $\O$ and issue $p$. 
(c) For any $\O \in \AB_c$ and $p \in \Atoms$, $u(\O)(p) = \ast$ if and only if $\O(p) = \O(\neg p)$, that is, an even number of voters vote and the group is split in half. 
(d) $F = \arg\min_G u(G)(p)$, for all $p \in \Atoms$.
\end{proof}
That is, the quota rule(s) corresponding to the majority rule (Example \ref{example:maj}) is precisely the rule that minimizes undecisiveness.

We can now state and prove the characterization result:
\begin{theorem}\label{thm:quotarules}
Let $F: (\AB_{c})^\N \to \AB$ be an aggregator for a given $\S$. The following holds: 
\begin{enumerate}
\item $F$ is a quota rule if and only if it is anonymous, independent, monotonic, and responsive; 
\item $F$ is a uniform quota rule if and only if it is a neutral quota rule;
\item $F$ is a symmetric quota rule if and only if it is an unbiased quota rule; 
\item $F$ is the majority rule $\maj$ if and only if it is a uniform symmetric quota rule which minimizes undecisiveness. 
\end{enumerate}
\end{theorem}

\begin{proof}
\smallskip
\fbox{Claim 1}
Left-to-right: Easily checked.
Right-to-left: Let $F$ be an anonymous, independent, monotonic, and responsive aggregator. By \textit{anonymity} and \textit{independence}, for any $p\in\Atoms$, and any $\O\in\AB_c$, the only information determining the value of $F(O)(p)$ are the integers $|\O(p)|$ and $|\O(\neg p)|$. 

By \textit{responsiveness}, there exists a non-empty set of profiles $S^\1=\{\O\in\AB|F(\O)(p)=\1\}$. Pick $\O$ to be any profile in $S^\1$ with a minimal value of $\frac{|\O(p)|}{|\O(\pm p)|}$ and call this value $q_\1$. Now let $\O'$ be any profile such that $\O'=_{-i}\O$ and $\frac{|\O'(p)|}{|\O'(\pm p)|}>q_\1$. This implies that $O_i(p)=\0$ and $O'_i(p)=\1$. By \textit{monotonicity}, it follows that $F(\O')(p)= \1$. 
By iterating this argument a finite number of times we conclude that whenever $\frac{|\O(p)|}{|\O(\pm p)|} \geq q_\1$, we have that $F(\O)(p)=\1$. 
Given that $q_\1$ was defined as a minimal value, we conclude also that if $F(\O)(p)=1$, then $\frac{\O(p)}{\O(p^\pm)}\geq q_\1$. The argument for $q_\0$ is identical.

\fbox{Claims 2 \& 3} follow straightforwardly from the definitions of uniform quota rule (Definition \ref{def:quota}) and of neutrality (Definition \ref{def:properties}) and, respectively, from the definitions of symmetric quota rules (Definition \ref{def:quota}) and of unbiasedness (Definition \ref{def:properties}) .

\fbox{Claim 4}
Left-to-right. Recall that $\maj$ is defined by quota $\frac{1}{2}< q_\1 =q_\0 \leq \frac{1}{2}+\frac{1}{|N|}$ (Example \ref{example:maj}). It is clear that $\maj$ is uniform and symmetric. The claim then follows by Lemma \ref{lemma:min}. 
Right-to-left. By Lemma \ref{lemma:min} if an aggregator minimizes undecisiveness then its quota are set as $\frac{1}{2}< q_\1 =q_\0 \leq \frac{1}{2}+\frac{1}{|N|}$. These quota define $\maj$ (Example \ref{example:maj}).
\end{proof}
By the above theorem and Fact \ref{fact:rat}, it follows that, on simple agendas, majority is the only rational aggregator which is also responsive, anonymous, systematic and monotonic.


\subsection{Impossibility in Binary Aggregation with Abstentions}

The following is a well-know impossibility result concerning binary aggregation with abstentions:
\begin{theorem}[\cite{Dokow_2010,Dietrich_2007}]\label{thm:imp.agg.abst}
Let $\S$ be a BA structure whose agenda is path connected and evenly negatable. Then if an aggregator $F: (\AB_{c})^\N \to \AB$ is independent, unanimous and collectively rational, then it is oligarchic.
\end{theorem}
We will use this result to illustrate how impossibility results from binary aggregation with abstentions apply to delegable proxy voting on binary issues.


\section{Liquid Democracy as Binary Aggregation} \label{sec:proxy}

In this section we provide an analysis of liquid democracy by embedding it in the theory of binary aggregation with abstentions presented in the previous section. To the best of our knowledge, this is the first attempt at providing an analysis of delegable proxy voting using social-choice theoretic tools, with the possible exception of \cite{greenarmytage_delegable}. 

In what follows we will often refer to delegable proxy voting/aggregation simply as proxy voting/aggregation.

\subsection{Binary Aggregation via Delegable Proxy}

In binary aggregation with proxy, agents either express an acceptance/rejection opinion or \emph{delegate} such opinion to a different agent. 

\subsubsection{Proxy Opinions and Profiles}

Let a BA structure $\S$ be given and assume for now that $\gamma = \top$, that is, all issues are logically independent. An opinion $O: \Atoms \to \set{\0,\1} \cup \N$ is an assignment of either a truth value or another agent to each issue in $\Atoms$, such that $O_i(p) \neq i$ (that is, self-delegation is not an expressible opinion).
We will later also require proxy opinion to be individually rational, in a precise sense (Section \ref{sec:indirat}). 
For simplicity we are assuming that abstention is not a feasible opinion in proxy voting, but that is an assumption that can be easily lifted in what follows.

We call functions of the above kind {\em proxy opinions} to distinguish them from standard (binary) opinions, and we denote by $\mathcal{P}$ the set of all proxy opinions, $\mathcal{P}_c$ the set of all consistent proxy opinions, $\mathcal{P}^\N$ being the set of all proxy profiles.  


\subsubsection{Delegation Graphs}

Each profile $\O$ of proxy opinions ({\em proxy profile} in short) induces a delegation graph $G^\O = \langle \N, \set{R_p}_{p \in \I}\rangle$ where for $i, j \in \N$:
\begin{align}
iR_p j & \IFF
\left\{
\begin{array}{ll}
O_i(p) = j & \mbox{if $i \neq j$} \\
O_i(p) \in \set{\0, \1} & \mbox{otherwise}
\end{array}
\right.
\end{align}
The expression $iR_pj$ stands for ``$i$ delegates her vote to $j$ on issue $p$''. 
Each $R_p$ is a so-called functional relation. It corresponds to the graph of an endomap on $\N$. So we will sometimes refer to the endomap $r_p: \N \to \N$ of which $R_p$ is the graph. Relations $R_p$ have a very specific structure and can be thought of as a set of trees whose roots all belong to cycles (possibly loops).

The weight of an agent $i$ w.r.t. $p$ in a delegation graph $G^\O$ is given by its indegree with respect to $R^*_p$ (i.e., the reflexive and transitive closure of $R_p$):\footnote{
We recall that the reflexive transitive closure $R^*$ of a binary relation $R \subseteq \N^2$ is the smallest reflexive and transitive relation that contains $R$.
} 
 $w^\O_i(p) = |\set{j \in \N \mid j R^*_p i}|$. This definition of weight makes sure that each individual carries the same weight, independently of the structure of the delegation graph. Alternative definitions of weight are of course possible and we will come back to this issue later.\footnote{See also footnote \ref{footnote:weight} below.} 
 
 For all $p \in \I$, we also define the function $g_p: \N \rightarrow \wp(\N)$ such that $g_p(i) = \set{j \in \N \mid j R^*_p i \AND \nexists k: jR_p k}$. The function associates to each agent $i$ (for a given issue $p$), the (singleton consisting of the) last agent reachable from $i$ via a path of delegation on issue $p$, when it exists (and $\emptyset$ otherwise). Slightly abusing notation we will use $g_p(i)$ to denote an agent, that is, the {\em guru} of $i$ over $p$ when $g_p(i) \neq \emptyset$. If $g_p(i) = \set{i}$ we call $i$ a {\em guru} for $p$. Notice that $g_p(i) = \set{i}$ iff $r_p(i) = i$, that is, $i$ is a guru of $p$ iff it is a fixpoint of the endomap $r_p$.
 
If the delegation graph $G^\O$ of a proxy profile $\O$ is such that, for some $R_p$, there exists no $i \in N$ such that $i$ is a guru of $p$, we say that graph $G^\O$ (and profile $\O$) is {\em void} on $p$. Intuitively, a void profile on $p$ is a profile where no voter expresses an opinion on $p$, because every voter delegates her vote to somebody else.

Given a BA structure $\S$, a proxy aggregation rule (or proxy aggregator) for $\S$ is a function $\pv:\mathcal{P}^\N\to\AB$ that maps every proxy profile to one collective incomplete opinion. As above, $\pv(\O)(p)$ denotes the outcome of the aggregation on issue $p$.

\subsubsection{Proxy Aggregators}

The most natural form of voting via delegable proxy is a proxy version of the majority rule we discussed in Section \ref{sec:preliminaries}:\footnote{On the importance of majority decisions in the current implementation of liquid democracy by liquid feedback cf. \cite[p.106]{liquid_feedback}.}
\begin{align}\label{eq:proxymajast}
\pv_{\maj}(\O)(p) =
\begin{cases}
 \1 &  \mbox{ if }  \sum_{i \in \O(p)} w^\O_i(p) > \sum_{i \in \O(\neg p)}w^\O_i(p) \\
 \0 &  \mbox{ if }  \sum_{i \in \O(\neg p)} w^\O_i(p) > \sum_{i \in \O(p)}w^\O_i(p) \\
 \ast & \mbox{ otherwise }
\end{cases}
\end{align}
That is, an issue is accepted by proxy majority in profile $\O$ if the sum of the weights of the agents who accept $p$ in $\O$ exceeds the majority quota, it is rejected if the sum of the weights of the agents who reject $p$ in $\O$ exceeds the majority quota, and it is undecided otherwise. It should be clear that $\sum_{i \in \O(p)} w^\O_i(p) = |\{i \in \N | O_{g_i}(p)= \1 \}|$ (and similarly for $\neg p$), that is, the sum of the weights of the gurus accepting (rejecting) $p$ is precisely the cardinality of the set of agents whose gurus accept (reject) $p$. 

In general, it should be clear that for any quota rule $F: \AB_c \to \AB$ a proxy variant $\pv_F$ of $F$ can be defined via an obvious adaptation of \eqref{eq:proxymajast}.

\medskip

To fix intuitions further about proxy voting it is worth discussing another example of aggregator, {\em proxy dictatorship}. It is defined as follows, for a given $d \in \N$ (the dictator) any proxy profile $\O$ and issue $p$: 
\begin{align}\label{eq:proxyd}
\pv_{d}(\O)(p) =
\begin{cases}
 \O_{g_p(d)} &  \mbox{ if }  g_p(d) \neq \emptyset \\
 \ast & \mbox{ otherwise }
\end{cases}
\end{align} 
That is, in a proxy dictatorship, the collective opinion is the opinion of the guru of the dictator, when it exists, and it is undefined otherwise.


\subsection{Two Issues of Delegable Proxy}

\subsubsection{Cycles and Abstentions} \label{sec:proxyabs}

It should be clear from the definition of proxy aggregators like $\pv_\maj$, that such aggregators rely on the existence of gurus in the underlying delegation graphs. If the delegation graph $R_p$ on issue $p$ contains no guru, then the aggregator has access to no information in terms of who accepts and who rejects issue $p$. To avoid bias in favor of acceptance or rejection, such situations should therefore result in an undecided collective opinion. That is for instance the case of $\pv_\maj$. However, such situations may well be considered problematic, and the natural question arises therefore of how likely they are, at least in principle.
\begin{fact} \label{fact:cycles}
Let $\A$ be a BA structure where $\gamma = \top$ (i.e., issues are independent).
If each proxy profile is equally probable (impartial culture assumption), then the probability that, for each issue $p$, the delegation graph $R_p$ has no gurus tends to $\frac{1}{e^2}$ as $n$ tends to infinity. 
\end{fact}
\begin{proof}
The claim amounts to computing the probability that a random proxy profile $\O$ induces a delegation graph $R_p$ that does not contain gurus (or equivalently, whose endomap $r_p: \N \to \N$ has no fixpoints) as $n$ tends to infinity.
Now, for each agent $i$, the number of possible opinions on a given issue $p$ (that is, functions $O: \set{p} \to \set{\0,\1} \cup \N$) is $|(\N \backslash \set{i}) \cup \set{\0,\1}| = n + 1$ (recall $i$ cannot express ``$i$'' as an opinion). The number of opinions in which $i$ is delegating her vote is $n - 1$. So, the probability that a random opinion of $i$ about $p$ is an opinion delegating $i$'s vote is $\frac{n-1}{n+1}$. Hence the probability that a random profile consists only of delegated votes (no gurus) is $(\frac{n-1}{n+1})^n$.
The claimed value is then established through this series of equations:
\begin{align*}
\lim_{n \to \infty} \left(\frac{n-1}{n}\right)^n & = \lim_{n \to \infty} \left(\frac{n}{n+2}\right)^n \\
& = \lim_{n \to \infty} \left(\frac{1}{\frac{n+2}{n}}\right)^n \\
& = \lim_{n \to \infty} \left(\frac{1}{1 + \frac{2}{n}}\right)^n \\
& = \lim_{n \to \infty} \left(\frac{1}{(1 + \frac{2}{n})^n}\right) \\
& = \frac{1}{\lim_{n \to \infty}(1 + \frac{2}{n})^n}\\
& = \frac{1}{e^2}
\end{align*}
This completes the proof.
\end{proof}
Now contrast the above simple fact with the probability that all agents abstain on an issue when each voter either expresses a $\1$ or $\0$ opinion or abstains (that is, the binary aggregation with abstentions setting studied earlier).
In that case the probability that everybody abstains tends to $0$ as $n$ tends to infinity.
 
Fact \ref{fact:cycles} should obviously not be taken as a realistic estimate of the effect of cycles on collective abstention, as the impartial culture assumption is a highly idealized assumption.
Election data should ideally be used to assess whether delegation cycles ever lead large parts of the electorate to 'lose their vote', possibly together with refinements of the above argument that take into consideration realistic distributions on proxy profiles, and therefore realistic delegation structures.
Nonetheless, Fact \ref{fact:cycles} does flag a potential problem of cyclical delegations as sources of abstention which has, to the best of our knowledge, never been discussed. The mainstream position on cyclical delegations \cite[Section 2.4.1]{liquid_feedback} is:\footnote{Cf. also \cite{Behrens15}.} 
\begin{quote}
``The by far most discussed issue is the so-called circular delegation
problem. What happens if the transitive delegations lead to
a cycle, e.g. Alice delegates to Bob, Bob delegates to Chris, and
Chris delegates to Alice? Would this lead to an infinite voting
weight? Do we need to take special measures to prohibit such a
situation? In fact, this is a nonexistent problem: A cycle only exists as long as there is no activity in the cycle in which case the cycle has no effect. As already explained [\ldots], as soon as somebody casts a vote, their (outgoing) delegation will be suspended. Therefore, the cycle naturally disappears before it is used. In our example: If Alice and Chris decide to vote, then Alice will no longer delegate to Bob, and Chris will no longer delegate to Alice [\ldots]. If only Alice decides to vote, then only Alice's delegation to Bob is suspended and Alice would use a voting weight of 3. In either case the cycle is automatically resolved and the total voting weight used is 3.''
\end{quote}

We will discuss later (Section \ref{sec:diffusion}) a possible approach to mitigate this issue by suggesting a different interpretation of liquid democracy in terms of influence rather than delegation.

\subsubsection{Individual \& Collective Rationality} \label{sec:indirat}

In our discussion so far we have glossed over the issue of logically interdependent issues and collective rationality. The reason is that under the delegative interpretation of liquid democracy developed in this section individual rationality itself appears to be a more debatable requirement than it normally is in classical aggregation. 

A proxy opinion $O_i$ is {\em individually rational} if the set of formulas
\begin{align}
\set{\gamma} \cup \set{p \in \I \mid O_{g_p(i)}(p) = \1} \cup \set{\neg p \in \I \mid O_{g_p(i)}(p) = \0} \label{eq:ir}
\end{align}
is satisfiable (consistency), and if whenever \eqref{eq:ir} entails $\pm p$, then $\pm p$ belongs to it (closedness).\footnote{Cf. the definition of individual opinions in Section \ref{sec:preliminaries}.} That is, the integrity constraint $\gamma$ is consistent with $i$'s opinion on the issues she does not delegate on, and the opinions of her gurus (if they exist), and those opinions, taken together, are closed under logical consequence (w.r.t. the available issues). 

The constraint in \eqref{eq:ir} captures, one might say, an idealized way of how delegation works: voters are assumed to be able to check or monitor how their gurus are voting, and always modify their delegations if an inconsistency arises. The constraint remains, however, rather counterintuitive under a delegative interpretation of proxy voting. Aggregation via delegable proxy has at least the potential to represent individual opinions as irrational (inconsistent and/or not logically closed). 

\medskip

Like in the case of delegation cycles we will claim that the interpretation of liquid democracy in terms of influence to be developed in Section \ref{sec:diffusion}, rather than in terms of delegation, makes individual rationality at least as defensible as in the classical case.


\subsection{Embedding in Binary Aggregation with Abstentions}

\subsubsection{One man---One vote}

Aggregation in liquid democracy as conceived in \cite{liquid_feedback} should satisfy the principle that the opinion of every voter, whether expressed directly or through proxy, should be given the same weight:
\begin{quote}
``[\ldots] in fact every eligible voter has still exactly one vote [\ldots] unrestricted transitive delegations are an integral part of Liquid Democracy. [\ldots] Unrestricted transitive delegations are treating delegating voters and direct voters equally, which is most democratic and empowers those who could not organize themselves otherwise'' \cite[p.34-36]{liquid_feedback}
\end{quote}
In other words, this principle suggests that aggregation via delegable proxy should actually be `blind' for the specific type of delegation graph. Making this more formal, we can think of the above principle as suggesting that the only relevant content of a proxy profile is its translation into a standard opinion profile (with abstentions) via a function $t: \mathcal{P} \to \AB$ defined as follows: for any $i \in \N$ and $p \in \I$: $t(O_i(p)) = O_{g_p(i)}$ if $ g_p(i) \neq \emptyset$ (i.e., if $i$ has a guru for $p$), and $t(O_i(p)) = \ast$ otherwise. Clearly, if we assume proxy profiles to be individually rational, the translation will map proxy opinions into individually rational (consistent and closed) incomplete opinions. By extension, we will denote by $t(\O)$ the incomplete opinion profile resulting from translating the individual opinions of a proxy profile $\O$.

\medskip

The above discussion suggests the definition of the following property of proxy aggregators: a proxy aggregator $\pv$ has the {\bf one man--one vote property} (or is a one man---one vote aggregator) if and only if $pv = t \circ F$ for some aggregator $F: \AB_c \to \AB$ (assuming the individual rationality of proxy profiles).\footnote{It should be clear that not every proxy aggregator satisfies this property. By means of example, consider an aggregator that uses the following notion of weight accrued by gurus in a delegation graph. The weight $w(i)$ of $i$ is $\sum_{j \in R^*(i)} \frac{1}{\ell(i,j)}$ where $\ell(i,j)$ denotes the length of the delegation path linking $j$ to $i$. This definition of weight is such that the contribution of voters decreases as their distance from the guru increases. Aggregators of this type are studied in \cite{Boldi_2011}. \label{footnote:weight}}

The class of one man---one vote aggregators can therefore be studied simply as the concatenation $t \circ F$ where $F$ is an aggregator for binary voting with abstentions,
as depicted in the following diagram:
\begin{center}
\begin{tikzpicture}
\node(P) at (0,0) {$\O$};
\node(O) at (2,0) {$t(\O)$};
\node(F) at (2, -2) {$F(t(\O))$};
\draw[->] (P) -- node[below]{$t$}     (O); 
\draw[->] (P) -- node[left]{$\pv_F$}  (F);
\draw[->] (O) -- node[right]{$F$}  (F);
\end{tikzpicture}
\end{center}
which gives us a handle to study a large class of proxy voting rules 

\begin{example}
Proxy majority $\pv_\maj$ \eqref{eq:proxymajast} is clearly a one man---one vote rule aggregator. It is easy to check that, for any proxy profile $\O$: $\pv_{\maj}(\O) = \maj(t(\O))$. The same holds for proxy dictatorship \eqref{eq:proxyd}. It is easy to see that proxy dictatorship $\pv_d$ is such that for any proxy profile $\O$: $\pv_{d}(\O) = d(t(\O))$, where $d$ is the standard dictatorship (of $d \in \N$).
\end{example}

It follows that for every proxy aggregator $\pv_F = t \circ F$ the axiomatic machinery developed for standard aggregators can be directly tapped into. 
Characterization results then extend effortlessly to proxy voting, again providing a strong rationale for the use of majority in proxy aggregation: 
\begin{fact}[Characterization of proxy majority]\label{thm:proxy.quotarules}
A one man---one vote proxy aggregator $\pv = t \circ F$ for a given $\S$ is proxy majority $pv_\maj$ iff $F$ is anonymous, independent, monotonic, responsive, neutral and minimizes undecisiveness.
\end{fact}
\begin{proof}
This follows from the definition of $t$ and Theorem \ref{thm:quotarules}. 
\end{proof}
It follows that on simple agendas and assuming the individual rationality of proxy profiles, proxy majority is the only rational aggregator which is anonymous, independent, monotonic, responsive, neutral and minimizes undecisiveness.

\subsubsection{Impossibility}

Similarly, there are many ways in which pursue the opposite embedding, from standard aggregation into proxy voting. For example, we can define a function $s: \mathcal{O}_c \to \mathcal{P}_c$ from opinion profiles to individually rational proxy profiles as follows. For a given opinion profile $\O$, and issue $p \in \I$ consider the set $\set{i\in \N \mid O_i(p) = \ast}$ of individuals that abstain in $\O$ and take an enumeration $1, \ldots, m$ of its elements, where $m = |\set{i\in \N \mid O_i(p) = \ast}|$. The function is defined as follows, for any $i \in \N$ and $p \in \I$: $s(O_i(p)) = O_i(p)$ if $O_i(p) \in \set{\0,\1}$, $s(O_i(p)) = i+1 \mod m$, otherwise.\footnote{Notice that since self-delegation (that is, $O_i(p) = i$) is not feasible in proxy opinions, this definition of $s$ works for profiles where, on each issue, either nobody abstains or at least two individuals abstain. A dummy voter can be introduced for that purpose.} 
A translation of this type allows to think of standard aggregators $F: \AB_c \to \AB$ as the concatenation $s \circ \pv$, for some proxy aggregator $\pv$:
\begin{center}
\begin{tikzpicture}
\node(O) at (2,0) {$\O$};
\node(Q) at (4,0) {$s(\O)$};
\node(F) at (2, -2) {$\pv(s(\O))$};
\draw[->] (Q) -- node[right]{$\pv$}  (F);
\draw[->] (O) -- node[right]{$F_\pv$}  (F);
\draw[->] (O) -- node[below]{$s$}  (Q);
\end{tikzpicture}
\end{center}

\begin{fact} \label{thm:imp.proxy}
Let $\S$ be such that its agenda is path connected and evenly negatable. For any proxy aggregator $\pv$, if $s \circ \pv$ is independent, unanimous and collectively rational, then it is oligarchic.
\end{fact}
\begin{proof}
It follows directly from the definition of $s$ and Theorem \ref{thm:imp.agg.abst}.
\end{proof}


\subsection{Section Summary}

The section has provided a very simple model of delegable proxy voting within the framework of binary aggregation. This has allowed us to put liquid democracy in perspective with an established body of results in the social choice theory tradition, and highlight two of its problematic aspects, which have so far gone unnoticed: the effect of cycles on collective indecisiveness, and the issue of preservation of individual rationality under delegable proxies.

An independent, purpose-built axiomatic analysis for liquid democracy focused on its more characteristic features (like the one man---one vote property) is a natural line of research, which we do not pursue here.


\section{Liquid Democracy as Binary Opinion Diffusion} \label{sec:diffusion}

Proxy voting can also be studied from a different perspective. Imagine a group where, for each issue $p$, each agent copies the $\0,\1$ opinion of a unique personal ``guru''. Imagine that this group does so repeatedly until all agents (possibly) reach a stable opinion. These new stable opinions can then be aggregated as the `true' opinions of the individuals in the group. The collective opinion of a group of agents who either express a $\0,\1$ opinion or delegate to another agent is (for one man---one vote proxy aggregators) the same as the output obtained from a vote where each individual has to express a $\0,\1$ opinion but gets there by copying the opinion of some unique ``guru'' (possibly themselves). 
In this perspective, a proxy voting aggregation can be assimilated to a (converging) process of opinion formation. 

The above interpretation of liquid democracy is explicitly put forth in \cite{liquid_feedback}:
\begin{quote}
``While one way to describe delegations is the transfer of voting weight to another person, you can alternatively think of delegations as automated copying of the ballot of a trustee.
While at assemblies with voting by a show of hands it is naturally possible to copy the vote of other people, in Liquid Democracy this becomes an intended principle'' \cite[p. 22]{liquid_feedback}.
\end{quote}
The current section develops an analysis of this interpretation, and highlights some of its advantages over the delegation-based one studied earlier.


\subsection{Binary aggregation and binary influence}

The section develops a very simple model of binary influence based on the standard framework of binary aggregation (see Appendix \ref{appendix:graph} for a concise presentation). For simplicity, in this section we assume agents are therefore not allowed to abstain, although this is not a crucial assumption for the development of our analysis.


\subsubsection{DeGroot Processes and Opinion Diffusion}

In \cite{Degroot_1974}, DeGroot proposes a simple model of step-by-step opinion change under social influence. The model combines two types of matrices. Assuming a group of $n$ agents, a first $n\times n$ matrix represents the weighted influence network (who influences whom and how much), and a second $n \times m$ matrix represents the probability assigned by each agent to each of the $m$ different alternatives. Both the agents' opinion and the influence weights are taken within $[0,1]$ and are (row) stochastic (each row sums up to $1$). Given an opinion and an influence matrix, the opinion of each agent in the next time step is obtained through linear averaging.

Here we focus on a specific class of opinion diffusion processes in which opinions are binary, and agents are influenced by exactly one influencer, possibly themselves, of which they copy the opinion. 
The model captures a class of processes which lies at the interface of two classes of diffusion models that have remained so far unrelated: the stochastic opinion diffusion model known as DeGroot's \cite{Degroot_1974}, and the more recent propositional opinion diffusion model due to \cite{Grandi:2015:POD:2772879.2773278}. The diffusion processes underpinning liquid democracy---which we call here Boolean DeGroot processes (BDPs)---are the $\set{0,1}$ special case of the DeGroot stochastic processes and, at the same time, the special case of propositional opinion diffusion processes where each agent has access to the opinion of exactly one neighbor (cf. Figure \ref{figure:intersection}).

\begin{figure}[t]
\begin{center}
\includegraphics[width=0.7\columnwidth]{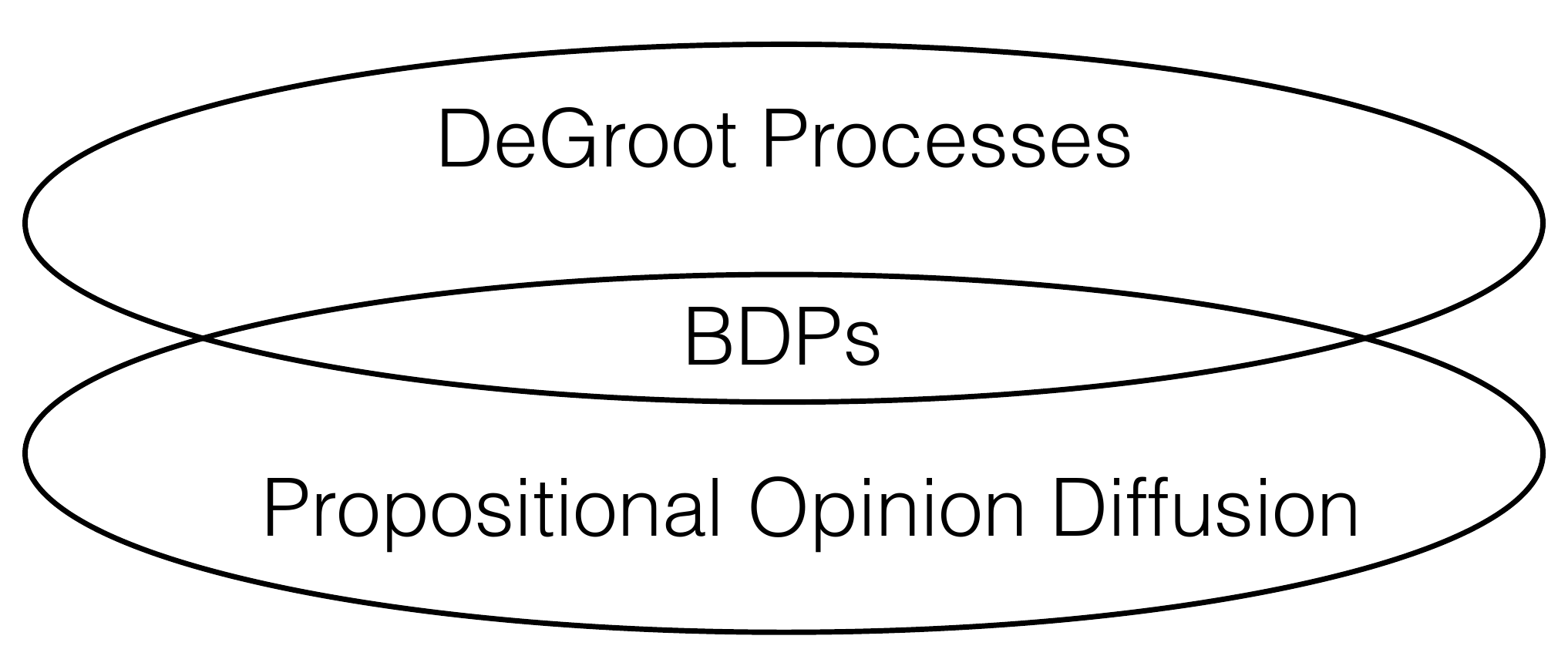}
\caption{{{BDPs lie in the intersection of DeGroot processes and propositional opinion diffusion processes.\label{figure:intersection}%
}%
}}
\end{center}
\end{figure}

\subsubsection{Boolean DeGroot processes}

Here we focus on the Boolean special case of a DeGroot process showing its relevance for the analysis of liquid democracy. Opinions are defined over a BA structure, and hence are taken to be binary. Similarly, we take influence to be modeled by the binary case of an influence matrix. Influence is of an ``all-or-nothing'' type and each agent is therefore taken to be influenced by exactly one agent, possibly itself. The graph induced by such a binary influence matrix (called \emph{influence graph}) is therefore a structure $G = \tuple{\N, R}$ where $R \subseteq \N^2$ is a binary relation where $i R j$ is taken here to denote that ``$i$ is influenced by $j$''. Such relation is serial ($\forall i\in \N, \exists j \in \N: i R j$) and functional ($\forall i,j,k \in \N$ if $i R j$ and $i R k$ then $j = k$). So each agent $i$ has exactly one successor (the influencer), possibly itself, which we denote $R(i)$. It should be clear that influence graphs are the same sort of structures we studied earlier in Section \ref{sec:proxy} under the label 'delegation graph'.

An \emph{influence profile} $\G=(G_1,\dots,G_m)$ records how each agent is influenced by each other agent, with respect to each issue $p \in \Atoms$. Given a profile $\G$ the i$^{\mathit{th}}$ projection $G_i$ denotes the influence graph for issue $p_i$, also written $G_p$.

\medskip

So let us define the type of opinion dynamics driving BDPs:
\begin{definition}
[BDP] \label{def:BDP}
Now fix an opinion profile $\O$ and an influence profile $\G$. Consider the stream $\O^0, \O^1, \ldots, \O^n, \ldots$ of opinion profiles recursively defined as follows:
\begin{itemize}
\item Base: $\O_0 := \O$
\item Step: for all $i \in \N$, $p\in \Atoms$, $O_i^{n+1}(p) := O^{n}_{R_p(i)}(p)$.
\end{itemize}
where $G_p = \tuple{\N, R_p}$.
We call processes defined by the above dynamics \emph{Boolean DeGroot processes} (BDPs).
\end{definition}
It should be clear that the above dynamics is the extreme case of linear averaging applied on binary opinions and binary influence.

As noted above, BDPs are also the special case of processes that have recently been proposed in the multi-agent systems literature as \emph{propositional opinion diffusion} processes \cite{Grandi:2015:POD:2772879.2773278}, i.e., cases where 1) the aggregation rule is the unanimity rule (an agent adopts an opinion if and only if all her influencers agree on it), and 2) each agent has exactly one influencer.  We will come back to the link with propositional opinion diffusion in some more detail later in Section \ref{sec:coloring}.


\subsection{Convergence of BDPs} \label{sec:convergence}

When do the opinions of a group of individuals influencing each other stabilize? Conditions have been given, in the literature, for the general paradigms of which BDPs are limit cases. This section introduces the necessary graph-theoretic notions and briefly recalls those results before giving a characterization of convergence for BDPs.

\subsubsection{Preliminaries}

We start with some terminology.
We say that the stream of opinion profiles $\O^0, \O^1, \ldots, \O^n, \ldots$ {\em converges} if 
there exists $n \in \mathbb{N}$ such that for all $m\in \mathbb{N}$, if $m\geq n$, then $\O^m = \O^n$.

We will also say that a stream of opinion profiles converges {\em for issue} $p$ if 
there exists $n \in \mathbb{N}$ such that, for all $m\in \mathbb{N}$, if $m\geq n$, then $\O^m (p) = \O^n(p)$.
Given a stream of opinion profiles starting at $\O$ we say that agent $i \in \N$ stabilizes in that stream for issue $p$ if there exists $n \in \mathbb{N}$ such that $O^n_i(p) = O^{m}_i(p)$ for any $m > n$. So a BDP on influence graph $\G$ starting with the opinion profile $\O$ is said to converge if the stream $\O^0, \O^1, \ldots, \O^n, \ldots$ generated according to Definition \ref{def:BDP} where $\O = \O^0$ converges. Similarly, A BDP is said to converge for issue $p$ if its stream converges for $p$, and an agent $i$ in the BDP is said to stabilize for $p$ if it stabilizes for $p$ in the stream generated by the BDP.

\medskip

Notice first of all that influence graphs have a special shape:\footnote{Please consult Appendix \ref{appendix:graph} for the relevant terminology from graph theory.}
\begin{fact}
\label{fact:uniquecycle}
Let $G$ be an influence graph and $C$ be a connected component of $G$. 
Then $C$ contains exactly one cycle, and the set of nodes in the cycle is closed. 
\end{fact}

\begin{proof}
Assume that $C$ does not contain any cycle. Since $\N$ is finite and since no path can repeat any node, any path in $C$ is finite too. Let $i$ be the last element of (one of) the longest path(s) in $C$. Then $i$ does not have any successor, which contradicts seriality. So $C$ contains at least one cycle. 
Let $S$ be the set of nodes of a cycle in $C$. Assume that $S$ is not closed: for some $i\in S$ and $j\notin S$, $iRj$. Since $S$ is a cycle, there is also some $k\in S$, such that $iRk$, which contradicts functionality. Therefore, the nodes of any cycle in $C$ forms a closed set. 
Now assume that $C$ contains more than one cycle. Since the nodes of each cycle forms a closed set, there is no path connecting any node inside a cycle to any node in any other cycle, which contradicts connectedness. So $C$ contains a unique cycle, whose nodes form a closed set. 
\end{proof}
Intuitively, influence graphs of BDPs then look like sets of confluent chains aiming together towards common cycles.

\subsubsection{Context: convergence in DeGroot processes}

For the general case of DeGroot processes, an influence structure guarantees that any distribution of opinions will converge if and only if ``every set of nodes that is strongly connected and closed is aperiodic" \cite[p.233]{jackson08social}.
In the propositional opinion diffusion setting, sufficient conditions for stabilization have been given by \cite[Th. 2]{Grandi:2015:POD:2772879.2773278}: on influence structures containing cycles of size at most one (i.e, only self-loops), for agents using an aggregation function satisfying (ballot-)monotonicity and unanimity\footnote{Notice that the rule underpinning BDP, that is the `guru-copying' rule on serial and functional graphs, trivially satisfies those constraints.}, opinions will always converge in at most at most $k+1$ steps, where $k$ is the diameter of the graph.\footnote{A second sufficient condition for convergence is given by \cite{Grandi:2015:POD:2772879.2773278}: when agents use the unanimity aggregation rule, on irreflexive graphs with only vertex-disjoint cycles, such that for each cycle there exists an agent who has at least two influencers, opinions converge after at most $\N$ steps. Note that no BDP satisfies this second condition.} The results below show how BDPs are an interesting limit case of both DeGroot and propositional opinion diffusion processes.

\subsubsection{Two results}

It must be intuitively clear that non-convergence in a BDP is linked to the existence of cycles in the influence graphs. However, from the above observation 
(Fact~\ref{fact:uniquecycle}), we know that nodes in a cycle cannot have any influencers outside this cycle, and hence that cycles (including self-loops) can only occur at the ``tail'' of the influence graph.
Hence, if the opinions in the (unique) cycle do not converge, which can only happen in a cycle of length $\geq 2$, the opinions of the whole population in the same connected component will not converge. The above implies that for any influence graphs with a cycle of length $\geq 2$, there exists a distribution of opinions which loops. 
This brings us back to convergence result for general (not necessarily Boolean) DeGroot processes. Indeed, for functional and serial influence graphs, a closed connected component is aperiodic if and only if its cycle is of length $1$.

\begin{fact}
\label{fact:influence}
Let $\G$ be an influence profile. Then the following are equivalent:
\begin{enumerate}
\item The BDP converges for any opinion profile $\O$ on $\G$. 
\item For all $p\in \Atoms$, $G_{p}$ contains no cycle of length $\geq 2$.
\item For all $p\in \Atoms $, all closed connected components of $G_{p}$ are aperiodic.
\end{enumerate}
\end{fact}

\begin{proof}
\fbox{$2) \Rightarrow 1)$}
Let $p\in\Atoms$ and assume that $G_p$ contains no cycle of length $\geq 2$ and has diameter $k$. Let $C_p$ be a connected component of $G_p$. By Fact \ref{fact:uniquecycle}, $C_p$ contains a unique cycle, which, by assumption, is of length $1$. Hence, $C_p$ is aperiodic. Let $i$ be the node in the cycle. The opinion of $i$ will spread to all nodes in $C_p$ after at most $k$ steps. Therefore, all BDPs on $G$ will converge after at most $l$ steps, where $l$ is the maximum within the set of diameters of $G_p$ for all $p\in\Atoms$. 
\fbox{$1) \Rightarrow 3)$} We proceed by contraposition. 
Assume that for some $p\in\Atoms$, a connected component $C_p$ of $G_p$ contains a cycle of length $k\geq 2$. By \ref{fact:uniquecycle}, this cycle is unique, and therefore the greatest common divisor of the cycles lengths of $C_p$ is $k$, so $C_p$ is not aperiodic. 
Let $S$ be the set of nodes in the cycle. 
Let $\O$ be such that for some $i,j\in S$ with distance $d$ from $i$ to $j$, $O_i(p)\neq O_j(p)$. Then $\O_i(p)$ will not converge, but enter a loop of size $k$: for all $x\in\mathbb{N}$, $O^{x\times k}_i(p) \neq O^{(x\times k)+d}_i(p)$. Hence, $\O$ does not converge. 
\fbox{$3) \Rightarrow 2)$} Trivial.
\end{proof}
It is worth noticing that one direction (namely from $3$ to $1$) of the above result is actually a corollary of both the convergence result for DeGroot processes stated at the beginning of this section (cf. \cite{jackson08social}), and of a known convergence result for propositional opinion diffusion \cite[Th. 2]{Grandi:2015:POD:2772879.2773278}, also stated earlier.

\medskip

The above gives a characterization of the class of influence profiles on which \emph{all} opinion streams converge. But we can aim at a more general result, characterizing the class of pairs of opinion and influence profiles which lead to convergence:

\begin{theorem}
\label{theorem:opinion}
Let $\G$ be an influence profile and $\O$ be an opinion profile. Then the following statements are equivalent:
\begin{enumerate}
\item The BDP converges for $\O$ on $\G$.
\item For all $p\in \Atoms$, there is no set of agents $S\subseteq\N$ such that: $S$ is a cycle in $G_{p}$ and there are two agents $i,j\in S$ such that $O_i(p)\neq O_j(p)$.
\end{enumerate}
\end{theorem}

\begin{proof}
\fbox{$1) \Rightarrow 2)$} We proceed by contraposition.
Let $p\in\Atoms$, $S\subseteq\N$ be a cycle in $G_p$, $i,j\in S$, and $O_i(p)\neq O_j(p)$. Let $k$ be the length of the cycle and $d$ be the distance from $i$ to $j$. Then $O_i(p)$ will enter a loop of size $k$: for all $x\in\mathbb{N}$, $O^{x k}_i(p)\neq O^{x k+d}_i(p)$.
\fbox{$2) \Rightarrow 1)$}
Assume $S\subseteq\N$ be such that $S$ is a cycle in $G_p$, and for all $i,j\in S$, $O_i(p)=O_j(p)$. Then, for all $j\in S$, and all $x\in\mathbb{N}$, $O^{x}_j(p)=O_i(p)$ and for all $f\in\N\notin S$ with distance $d$ from $f$ to $i$, for all $x\in\mathbb{N}$, such that $x\geq d$, $O^{x}_f(p)=O_i(p)$. 
\end{proof}
This trivially implies that the class of opinion profiles which guarantees convergence for {\em any} influence profile, is the one where everybody agrees on everything already. Note that the only stable distributions of opinions are the ones where, in each connected component in $G$, all members have the same opinion, i.e, on BDPs, converging and reaching a consensus (within each connected component) are equivalent, unlike in the stochastic case. Moreover, for an influence profile where influence graphs have at most diameter $d$ and the smallest cycle in components with diameter $d$ is of length $c$, it is easy to see that if a consensus is reached, it will be reached in at most $d-c$ steps, which is at most $n-1$.

\medskip

Finally observe that Theorem \ref{theorem:opinion} subsumes Fact \ref{fact:influence}. If $G_p$ contains only cycles of length $1$ (second statement in Fact \ref{fact:influence}) then, trivially, no two agents in a cycle can disagree (second statement in Theorem \ref{theorem:opinion}).

\subsubsection{Liquid Democracy as a BDP}

We have seen (Section \ref{sec:proxy}) that each proxy profile $\O$ induces what we called a delegation graph $G^\O = \tuple{N, R_p}$ for each issue $p$. Delegation graphs are the same sort of structures we referred to in the current section as influence graphs. So each proxy profile $\O$ can be associated to a BDP by simply assigning random $\0$ or $\1$ opinions to each voter delegating her vote in $\O$. It is then easy to show that for each connected component $C$ of $G^\O$, if $C$ has a guru with opinion $x$, then that component stabilizes in the BDP on opinion $x$ for each assignment of opinions to the delegating agents in $\O$. Vice versa, if $C$ stabilizes on value $x$ in the BDP for each assignment of opinions to the delegating agents in $\O$, then $C$ has a guru whose opinion is $x$. This establishes a direct correspondence between voting with delegable proxy and Boolean deGroot processes. However, BDPs offer an interesting and novel angle on the issue of cyclical delegations, to which we turn now.

\subsubsection{Cycles}

As discussed earlier (Section \ref{sec:proxyabs}), cycles are a much discussed issue in liquid democracy. Its proponents tend to dismiss delegation cycles as a non-issue: since the agents forming a cycle delegate their votes, none of them is casting a ballot and the cycles get resolved essentially by not counting the opinions of the agents involved in the cycle \cite{liquid_feedback}. We stressed this solution as problematic in the `vote-delegation' interpretation of liquid democracy as it has the potential to discard large numbers of opinions.  The elimination of cycles not only hides to aggregation the opinions of the agents involved in cycles, but also the opinions of agents that may be linked to any of those agents by a delegation path. In other words information about entire connected components in the delegation graph may be lost.

We argue that the `vote-copying' interpretation of the system---formalized through BDPs---offers novel insights into possible approaches to cycles in delegable proxy. Theorems \ref{theorem:opinion} and \ref{theorem:mu} offer an alternative solution by showing that not all cycles are necessarily bad news for convergence: cycles in which all agents agree still support convergence of opinions, and therefore a feasible aggregation of opinions by proxy. This suggests that alternative proxy voting mechanisms could be designed based on opinion convergence behavior rather than on weighted voting.


\subsection{Excursus: unanimity and 2-colorability}\label{sec:coloring}

In the above, we have worked at the intersection of two models of opinion diffusion, the DeGroot model, and the propositional opinion diffusion model. However, there is more to say about how the two frameworks relate.

Let us take a brief detour towards a generalisation of BDPs corresponding to the case of propositional opinion diffusion with the unanimity rule, where agents can have several influencers and change their opinions only if all their influencers disagree with them. This means that we relax the functionality constraint on influence graphs. We will show how the two frameworks meet again: some non-stabilizing opinion cases under the unanimity rule correspond to a special class among the `semi-Boolean' cases of DeGroot processes where opinions are still binary but influence does not need to be.

We define the dynamics of opinions under the unanimity rule in the obvious way:

\begin{definition}
[UP]
Fix an opinion profile $\O$ and a (serial but non-necessarily functional) influence profile $\G$. Consider the stream $\O^0, \O^1, \ldots, \O^n, \ldots$ of opinion profiles recursively defined as follows:
\begin{itemize}
\item Base: $\O_0 := \O$
\item Step: for all $i \in \N$ and all $p \in \Atoms$:
\begin{align}
\O_i^{n+1}(p) & = \left\{
\begin{array}{ll}
O_i^{n}(p) & \mbox{if for some $j,k\in R_p(i)$,$O_j^{n}(p)\neq O_k^{n}(p)$ } \\
O_j^{n}(p) & \mbox{otherwise, where $j \in R_p(i)$ } 
\end{array}
\right.
\end{align}
\end{itemize}
where $G_p = \tuple{\N, R_p}$.
We call processes defined by the above dynamics \emph{Unanimity Processes} (UPs).
\end{definition}

We give a sufficient condition for non-convergence of UPs:

\begin{lemma}
\label{lemma:suff.UP}
Let $G$ be a (serial and non-necessarily functional) influence profile and $\O$ be an opinion profile, such that, for some $p\in \Atoms $, for all $i,j\in C$, where $C$ is a connected component of $G_p$: if $i\in R_p(j)$, then $O_i(p)\neq O_j(p)$. 
Then $\O$ does not converge in UP.
\end{lemma}

\begin{proof}
Let $G$ be a (serial and non-necessarily functional) influence profile, and $\O$ be an opinion profile, such that, for some $p\in \Atoms $, for all $i,j\in C$ with $C$ a connected component of $G_p$: if $i\in R_p(j)$, then $O_i(p)\neq O_j(p)$. Then, by definition of UPs, for all $i\in C$, $O^1_i(p)\neq O_i(p)$, and by repeating the same argument, for all $n\in\mathbb{N}$, $O^{n+1}_i(p)\neq O^n_i(p)$.
\end{proof}

Intuitively, the above condition for non-convergence corresponds to a situation of global maximal disagreement: \emph{all} agents (of a connected component) disagree with \emph{all} their influencers. 
Recall that a graph is properly $k$-colored if each node is assigned exactly one among $k$ colors and no node has a successor of the same color, and consider the two possible opinions on issue $p$ as colors. The above result can be reformulated in terms of proper $2$ colorings, as follows: if for some $p\in\Atoms$, $\O$ properly colors $G_p$, then $\O$ does not converge. In such a case, all agents will change their opinion on $p$ at every step, entering an oscillation of size $2$. So the maximal state of disagreement is the maximally unstable case of the dynamics. Note that this limit case of opinion distribution is yet another special case of DeGroot processes, another example within the intersection between the two frameworks of propositional opinion diffusion and DeGroot.

The possibility of such a distribution of opinions on $p$ relies on the influence graph $G_p$ being $2$-colorable, which is again a requirement about the lengths of its cycles: it is $2$-colorable if and only if it contains no cycle of odd length. However, non $2$-colorability is not a sufficient condition for convergence of UPs in general: a simple cycle of three agents, for instance, is not $2$-colorable but does not guarantee convergence either (as illustrated above with the convergence conditions for BDPs). 
Nevertheless, there is a class of influence profiles for which being $2$-colorable is a necessary condition of non-convergence of UPs, the \emph{symmetric} ones:

\begin{lemma}
\label{lemma:symm.opinionUP}
Let $\G$ be a symmetric (serial and non-necessarily functional) influence profile and $\O$ be an opinion profile. The following statements are equivalent:  
\begin{enumerate}
\item $\O$ converges in UP on $\G$;
\item For all $p\in\Atoms $, for all connected component $C$ of $G_p$, there are $i,j\in C$, such that $i\in R_p(j)$, and $O_i(p)= O_j(p)$, where $G_p = \tuple{\N, R_p}$.
\end{enumerate}
\end{lemma}

\begin{proof}
\fbox{$2) \Rightarrow 1)$} Assume that for any $p\in\Atoms$, for any connected component $C$ of $G_p$, there exist $i,j\in C$, such that $ R_p(j)$ and $O_i(p)= O_j(p)$. By definition of UP, this implies that $O_i(p)$ is stable, and that all agents with distance $\leq k$ will be stable after at most $k$ steps. \fbox{$1) \Rightarrow 2)$} This follows from Lemma~\ref{lemma:suff.UP}.     
\end{proof}

This means that opinions on a given $p$ will converge if and only if two agents influencing each other on $p$ already agree on it. We can therefore, as we did for BDPs, characterize the class of influence profiles for which all (symmetric) opinion profiles converge in UPs:

\begin{theorem}
\label{thm:symm.influenceUP}
Let $\G$ be a symmetric (serial and non-necessarily functional) influence profile. The following statements are equivalent:
\begin{enumerate}
\item All opinion profiles $\O$, converge in UPs on $\G$.
\item For all $p\in\Atoms$, and all connected components of $C \subseteq G_p$, $C$ is not $2$-colorable (contains cycle(s) of odd length), where $G_p = \tuple{\N, R_p}$. 
\end{enumerate}
\end{theorem}

\begin{proof}
\fbox{$2) \Rightarrow 1)$} Let $p\in\Atoms$ and $C$ be connected component of $G_p$ with diameter $k$. Let $C$ contain a cycle of length $c$, with $c$ odd. Let $\O$ be an arbitrary opinion profile. Since $c$ is odd, there exist $i,j\in S$ such that $j\in R_p(i)$ and $O_i(p)=O_j(p).$ By definition of UP, this implies that $O_i(p)$ is stable, and that all agents with distance $\leq k$ will be stable after at most $k$ steps. Hence, $\O$ converges. \fbox{$1) \Rightarrow 2)$} This follows from Lemma~\ref{lemma:symm.opinionUP}.     
\end{proof}

Note that, while the basic modal language cannot capture graph $2$-colorability, it can capture non $2$-colorability, and therefore capture the class of symmetric (serial and non-necessarily functional) influence profiles which guarantee convergence of UPs. We leave the detail out for space reason.

We have shown that, for UPs in general, convergence (in a connected component) is not guaranteed if it contains no odd cycles, and that symmetric UPs guarantee convergence as soon as they contain some odd cycle. However, containing an odd cycle is a very ``easy'' requirement for a real-life influence network to meet (it corresponds to a non-zero clustering coefficient). By contrast, recall that BDPs guarantee convergence (on {\em any} opinion profile) only when they contain only cycles of size $1$, which is a rather implausible requirement to be satisfied on real influence networks.


\subsection{BDPs on logically interdependent issues}

So far we have assumed the aggregation to happen on a set of issues without constraint (or rather with $\gamma = \top$). In this subsection we study what happens in the presence of a constraint $\gamma \neq \top$. BDPs on aggregation structures with constraints may lead individuals to update with logically inconsistent opinions. But the diffusion perspective whereby agents copy the opinion of trustees rather than delegating their voting right better lends itself to an assumption of individual rationality. 

The following processes are simple adaptations of BDPs where agents update their opinions only if the opinions of their influencers, on the respective issues, are consistent with the constraint.\footnote{Other update policies are of course possible. A recent systematic investigation of opinion diffusion on interconnected issues is \cite{Botan16}.}

\begin{definition}
Fix an opinion profile $\O$, an influence profile $\G$, and a  constraint $\gamma$. Consider the stream $\O^0, \O^1, \ldots, \O^n, \ldots$ of opinion profiles recursively defined as follows:
\begin{itemize}
\item Base: $\O_0 := \O$
\item Step: for all $i \in \N$, $p\in \Atoms$, 
\begin{align*}
O_i^{n+1}(p) := 
\left\{
\begin{array}{ll}
O^{n}_{R_p(i)}(p) & \mbox{if    }  \bigwedge_{p \in \Atoms} O^{n}_{R_p(i)}(p) \wedge \gamma  \mbox{    is consistent} \\
O_i^{n}(p) & \mbox{otherwise}
\end{array}
\right.
\end{align*}
\end{itemize}
where $G_p = \tuple{\N, R_p}$.
We call processes defined by the above dynamics \emph{individually rational} BDPs.
\end{definition}

Individually rational BDPs converge in some cases in which BDPs do not. There are cases in which there is disagreement in the cycles but the process still converges, because of the safeguard towards individual rationality built into the dynamics.


\begin{example}
Consider the following example. Let $\N = \set{1,2}$, $\I = \set{p,q}$ and $\gamma =\{p\leftrightarrow \neg q\}$. Let then $G = \tuple{N, \set{R_i}_{i \in \I}}$ be as follows: $1R_q1$, $2R_q2$, $1R_p2$ and $2R_p1$. Finally let $\O$ be such that $O_1(p) = O_2(q) = \1$, $O_2(p) = O_1(q) = \0$. Voters $1$ and $2$ form a non-unanimous cycle, but $\O$ is a stable opinion profile.
\end{example}
The example shows that direction $1)\Rightarrow 2)$ of Theorem \ref{theorem:opinion} does not hold for individually rational BDPs: some individually rational BDPs may stabilize even in the presence of disagreement within a cycle. Intuitively, the reason why this happens is that individually rational BDPs that stabilize even when disagreements occur within cycles do so because their cycles are not "synchronized". In the above example, given the constraint $p\leftrightarrow\neg q$, the only way to get stabilization starting from a situation respecting the constraint is to have a cycle of influence for $q$ which goes `in the opposite direction' from the one from $p$, all other cases would amount to violate the constraint.

\medskip

Beyond this simple example, we want to find out what happens with more complex constraints and what are the  conditions for individually rational BDPs to converge.  Let us first show that direction $2)\Rightarrow 1)$ of Theorem \ref{theorem:opinion} still holds, that is, individually rational BDPs without disagreement in their cycles always converge:

\begin{theorem}
\label{theorem:resistantsufficient}
Let $\G$ be an influence profile, $\O$ be an opinion profile, and $\gamma$ a constraint. Then the following holds: {\em if} for all $p\in \Atoms$, for all $S\subseteq\N$ such that $S$ is a cycle in $G_{p}$, and all $i,j\in S$: $O_i(p)=O_j(p)$, {\em then} the individually rational BDP for $\O$, $\G$ and $\gamma$ converges in at most $k$ steps, where $k\leq \max\{diam(G_p)|p\in P\}$.
\end{theorem}
\begin{proof}
Assume that for all $p\in \Atoms$, for all $S\subseteq\N$ such that $S$ is a cycle in $G_{p}$, for all $i,j\in S$: $O_i(p)=O_j(p)$.
Consider an arbitrary $i\in \N$. 
Let $k_i(p)$ be the distance from $i$ to the closest agent in a cycle of $G_p$, and let $k_i$ denote $max \{k_i(p)|p\in P\}$. We show that for any $k_i\in\mathbb{N}$, $O^{k_i}_ i$ is stable.
\begin{itemize}
\item If $k_i=0$: $i$ is its only infuencer, therefore $O^0_{i}$ is stable. 
\item If $k_i=n+1$: Assume that for all agents $j$ such that $k_j=n$, $O^{k_j}_ j$ is stable. This implies that all influencers of $i$ are stable. We need to consider the following cases:
\begin{enumerate}
\item If $\bigwedge_{p \in \Atoms} O^{m}_{R_p(i)}(p) \wedge \gamma$ is not consistent, then it will never be: $O^{n}_i$ is stable. 
\item If $\bigwedge_{p \in \Atoms} O^{m}_{R_p(i)}(p) \wedge \gamma$ is consistent, 
then $O^{n+1}_i$ is stable. 
\end{enumerate}
\end{itemize}
This completes the proof.
 \end{proof}
 
\subsection{Section Summary} 
 
In this section we studied a very simple class of opinion diffusion processes on networks (Boolean DeGroot processes, BDPs), which precisely capture the vote-copying behavior suggested by a standard interpretation of the liquid democracy system. Interestingly these processes lie at the interface of two so far unconnected network diffusion models: the well-known DeGroot processes---of which BDPs constitute the binary special case---and of propositional opinion diffusion processes---of which BDPs constitute the special case where the set of neighbors is a singleton. 
We established necessary and sufficient conditions for convergence, which can be captured in modal fixpoint logics as we will show in the next section. We argued that these results provide a novel angle on the issue of delegation cycles in liquid democracy.
 
There are a number of further questions concerning, especially, individually rational BDPs that we leave for future investigations: What are the necessary conditions for their stabilization? What opinions are reachable? And, in particular, when is a consensus reached? Finally, one could consider other types of influence policies than the one used in individually rational BDPs. For instance, agents may be allowed to `pass through' an inconsistent state at some point, in which case one can wonder under which conditions the process can still converge to a consistent state. Indeterministic policies would also make sense, where an agent confronted with inconsistent opinions from her influencers keeps one of the closest consistent opinions set, rather than not being influenced at all (cf. \cite{Botan16}).


\section{Fixpoint Logics for BDPs} \label{sec:logic}

In this section we show how a well-established logic for formal verification can be readily used to specify and reason about properties of BDPs, and in particular their convergence. The logic is the so-called $\mu$-calculus. This points to a so-far unexplored interface between fixpoint logics and models of opinion dynamics---like the DeGroot model and propositional opinion diffusion. The section moves some first steps in that direction along the lines of another recent work \cite{JvBoscillations}, where the $mu$-calculus, and extensions thereof, have been applied to the study of dynamical systems.

\subsection{Influence graphs as Kripke models}

We treat influence graphs as Kripke (multi-relational) models \cite{Seligmanetal:synthese,Christoff_2015}.

\begin{definition}
We call an {\em influence model} a tuple $\Model = \tuple{\N, \G, \O}$ where $\G=(G_{p_1},\dots,G_{p_m})$ is an influence profile, and $\O: \Atoms \longrightarrow 2^\N$ is an opinion profile over $\Atoms$, that is, a valuation function.
\end{definition}

One can therefore easily interpret a modal language over influence models, where modalities are interpreted on the accessibility relations in $\G$. That is, to each graph $G_p$ we associate modalities $\lbox{p}$ and $\ldia{p}$. We will give the details below, but let us immediately note that the class of (possibly infinite) influence graphs would then be characterized by the following properties, for any $p \in \Atoms$:
\begin{align}
\lbox{p} \phi \limp \ldia{p} \phi & & \mbox{(seriality)} \\
\ldia{p} \phi \limp \lbox{p} \phi & & \mbox{(functionality)}
\end{align}
More precisely, for any influence profile $\G=(G_{p_1},\dots,G_{p_m})$, formula $\lbox{p_i} \phi \limp \ldia{p_i} \phi$ (respectively, $\ldia{p_i} \phi \limp \lbox{p_i} \phi$) is valid in such graph---that is, true in any pointed influence model built on such graph---if and only if each $G_{p_i}$ consists of a serial (respectively, functional) relation.\footnote{These are known results from modal correspondence theory (cf. \cite{Blackburn_2001}).} Put otherwise, on serial and functional graphs the modal box and diamond are equivalent.

\subsection{Modal $\mu$-calculus}

\[
\L^\mu: \phi ::=  p \mid \bot  \mid \neg \phi \mid \phi \land \phi \mid \ldia{p} \phi \mid \mu p. \phi(p)
\]

The language of the $\mu$-calculus expands the basic modal language with a least fixpoint operator $\mu$. Here is the BNF of the language:

where $p$ ranges over $\Atoms$ and $\phi(p)$ indicates that $p$ occurs free in $\phi$ (i.e., it is not bounded by fixpoint operators) and under an even number of negations.\footnote{This syntactic restriction guarantees that every formula $\phi(p)$ defines a set transformation which preserves $\subseteq$, which in turn guarantees the existence of least and greatest fixpoints by the Knaster-Tarski fixpoint theorem (cf. \cite{Stirling_2001}).} In general, the notation $\phi(\psi)$ stands for $\psi$ occurs in $\phi$. The usual definitions for Boolean and modal operators apply. Intuitively, $\mu p. \phi(p)$ denotes the smallest formula $p$ such that $p \lequiv \phi(p)$. The greatest fixpoint operator $\nu$ can be defined from $\mu$ as follows: $\nu p. \phi(p) := \neg \mu p. \neg \phi( \neg p)$.

We interpret $\L^\mu$ on influence models as follows:

\begin{definition}
Let $\phi \in \L^\mu$. The satisfaction of $\phi$ by a pointed influence model $(\Model, i)$ is inductively defined as follows:
\begin{align*}
\Model, i \not\models \bot     &           \\
\Model, i \models p   \   &      \IFF            i \in \O(p),      \mbox{        for       }    p \in {\bf P}  \\
\Model, i \models \neg \phi     &      \IFF       i \not\in \true{\phi}_\Model  \\
\Model, i \models \phi_1 \wedge \phi_2     &      \IFF     i \in  \true{\phi_1}_\Model \cap  \true{\phi_2}_\Model \\
\Model, i \models  \ldia{p} \phi &     \IFF       i \in \{ j \mid \exists k: j G_p k \ \& \ k \in \true{\phi}_\Model \} \\
\Model, i \models  \mu p. \phi(p) &      \IFF      i \in \bigcap \{ X \in 2^\N \mid \true{\phi}_{\Model[p:=X]}  \subseteq  X    \}
\end{align*}
where $\true{\phi}_{\Model[p:=X]}$ denotes the truth-set of $\phi$ once $\O(p)$ is set to be $X$. As usual, we say that: $\phi$ is valid in a model $\Model$ iff it is satisfied in all points of $\Model$, i.e., $\Model \models \phi$; $\phi$ is valid in a class of models iff it is valid in all the models in the class.
\end{definition}
We list some relevant known results about $\K^{\mu}$. The logic has a sound and (weakly) complete axiom system \cite{Walukiewicz_2000}. The satisfiability problem of $\K^{\mu}$ is decidable \cite{Streett_1984}. The complexity of the model-checking problem for $\K^{\mu}$ is known to be in NP $\cap$ co-NP  \cite{Gr_del_1999}. It is known that the model-checking problem for a formula of size $m$ and alternation depth $d$ on a system of size $n$ can be solved by the natural fixpoint-approximation algorithm with (time) complexity of $O((m \cdot n)^{d+1})$ \cite{Emerson96}, where the alternation depth of a formula of $\L^\mu$ is the maximum number of $\mu/\nu$ alternations in a chain of nested fixpoint subformulas.\footnote{The reader is referred to, e.g. \cite{Emerson_2001}, for the precise definition.} Finally, the $\mu$-calculus is known to be invariant for bisimulation (cf. \cite{Blackburn_2001}). It is actually known to correspond to the bisimulation-invariant fragment of monadic second-order logic \cite{Janin_1996}.

\subsection{On the logic of convergence in BDPs}

Each stream of opinion profiles $\O^0, \O^1, \ldots, \O^n, \ldots$ corresponds to a stream of influence models $\Model^0, \Model^1, \ldots, \Model^n, \ldots$.

From the point of view of an influence model $\Model = \tuple{\N, \G, \O}$ the BDP dynamics of Definition \ref{def:BDP} can therefore be recast in terms of updates of the valuation function $\O$ as follows:

\begin{itemize}
\item Base: $\O^0 := \O$
\item Step: $\O^{n+1}(p) := \true{\lbox{p}p}_{\Model^n}$.
\end{itemize}
That is, the interpretation of $p$ at step $n+1$ is the interpretation of $\lbox{p}p$ at step $n$. Equivalently, the interpretation of $\neg p$ at step $n+1$ is the interpretation of $\lbox{p}\neg p$ at step $n$.

\begin{lemma}
\label{lemma:stable}
Let $\Model = \tuple{\N, \G, \O}$ be an influence model. The two following statements are equivalent:
\begin{enumerate}
\item $i \in \N$ is stable for $p$;
\item The pointed model $(\Model, i)$ satisfies:\footnote{Notice that $\pm p$ is used as a variable ranging over $\set{p, \neg p}$. Technically the above formula is to be read as a scheme for $\nu x. p \land \lbox{p} x$ and $\nu x. \pm \neg p \land \lbox{p} x$.}
\begin{align}
\Stb(p) := \nu x. \pm p \land \lbox{p} x 
\end{align}
\end{enumerate}
\end{lemma}

\begin{proof}
First of all observe that, by the semantics of the $\mu$-calculus, formula $\Stb(p)$ denotes the largest fixpoint of function $\pm p \land \lbox{p}(\cdot)$, that is, formula $\lbox{p^*} \pm p$ where $\lbox{p^*}$ is the modal box interpreted over the reflexive and transitive closure of $G_p$.
\fbox{$1) \Rightarrow 2)$} Assume that $i$ is stable for $p$ and suppose towards a contradiction that $\Model, i \not\models \Stb(p)$. By what said above, it follows that there exists a $j$ such that $\O_i (p) \neq \O_j(p)$ which is connected by a finite $G_p$ path to $i$. By the functionality of influence models and the dynamics of Definition \ref{def:BDP} then at some stage $n$ in the stream of opinion profiles it should hold that $\O^n_i(p) = \O_j(p)$, against the assumption that $i$ be stable for $p$.
\fbox{$2) \Rightarrow 1)$} Assume $\Model, i \models \Stb(p)$. By what said above, this implies that there exists no $j$ such that $\O_i (p) \neq \O_j(p)$ which is connected by a finite $G_p$ path to $i$. It follows that in the stream generated by the BDP dynamics $i$ cannot change its opinion, and hence it is stable.
\end{proof}

\begin{theorem}
\label{theorem:mu}
Let $\Model = \tuple{\N, \G, \O}$ be an influence model. The two following statements are equivalent:
\begin{enumerate}
\item $i \in \N$ stabilizes for issue $p \in \Atoms$;
\item The pointed model $(\Model, i)$ satisfies:
\begin{align}
\mu x. \Stb(p) \vee \lbox{p} x
\end{align}
\end{enumerate}
\end{theorem}

\begin{proof}
First of all observe that, by the semantics of the $\mu$-calculus $\mu x. \Stb(p) \vee \lbox{p} x$ denotes the smallest fixpoint of equation $x \lequiv \Stb(p) \vee \lbox{p} x$. By the Knaster-Tarski theorem and the fact that influence models are finite, we can compute such fixpoint as $\bigcup_{0 \leq n < \omega} \true{\Stb(p)^n}$ where $\true{\Stb(p)^0} = \true{\Stb(p) \vee \lbox{p} \bot}$ (notice that $\lbox{p} \bot \lequiv \bot$ on influence models) and $\true{\Stb(p)^{n+1}} = \true{\Stb(p) \vee \lbox{p}\Stb(p)^n}$. So, by Lemma \ref{lemma:stable} $i$ belongs to $\true{\mu x. \Stb(p) \vee \lbox{p} x}$ either $i$ is stable for issue $p$ or has access in a finite number of steps to a an agent who is stable for $p$.
\fbox{$1) \Rightarrow 2)$} Assume that $i$ stabilizes for issue $p \in \Atoms$. So there exists a stage $n$ in the stream of profiles generated through Definition \ref{def:BDP} at which $\O_i^n(p) = \O_i^{m}(p)$ for all $m > n$. By Lemma \ref{lemma:stable}, $\tuple{\N, \G, \O^n}, i \models \Stb(p)$. It follows that $i$ is connected through a finite $G_p$-path to an agent $j$ such that $\Model, j \models \Stb(p)$. By what established above we thus have that $\Model, i \models \mu x. \Stb(p) \vee \lbox{p} x$.
\fbox{$2) \Rightarrow 1)$} Assume $\Model, i \models \mu x. \Stb(p) \vee \lbox{p} x$. It follows that $i$ is connected through a finite $G_p$-path to an agent $j$ such that $\Model, j \models \Stb(p)$. By Lemma \ref{lemma:stable} $j$ is therefore stable and therefore $i$ will stabilize for $p$.
\end{proof}
So the formula that expresses the stabilization of the agents' opinions on one issue is $\mu x. \left(\nu y. \pm p \land \lbox{p} y \right) \vee \lbox{p} x$.
Informally, the theorem states that in a BDP an agent reaches a stable opinion if and only if it has an indirect influencer (linked by an influence path) whose all direct and indirect influencer have the same opinion. Notice that such formula has alternation depth $0$. So an off-the-shelf model-checking algorithm for the $\mu$-calculus can check stabilization in time $O(m \cdot n)$ with $n$ being the size of the model and $m$ the size of the formula.

Now confront this with the earlier Theorem \ref{theorem:mu}. Since the convergence of the BDP is equivalent to the stabilization of all agents on all issues $p$ (either on $p$ or $\neg p$), we have the following corollary:

\begin{corollary}
The BDP for an opinion profile $\O$ based on influence graph $\G$ converges if and only if
\begin{align}
\tuple{\N, \G, \O}, i \models U \left(\bigwedge_{p \in \Atoms} \mu x. \Stb(p) \vee \lbox{p} x \right)
\end{align}
for any agent $i \in \N$, where $U$ denotes the universal modality (cf. \cite{Blackburn_2001}).
\end{corollary}
So the above formula characterizes the property of convergence for a BDP. Since the process of voting in a liquid democracy system can be modeled by a BDP, the formula also characterizes precisely when voting by delegable proxy results in a $\1$ or $\0$ opinion on a given issue. 


\section{Conclusions} \label{sec:conclusions}

The paper has moved the first steps towards the development of theoretical foundations for the voting system of liquid democracy based on delegable proxy. 

We have pursued two lines of research linked to two interpretations commonly associated to the proxy character of liquid democracy: the delegation of voting right to trustees, vs. the copying of the votes of influencers. The first interpretation has led us to develop a simple model of liquid democracy based on the theory of binary and judgment aggregation. This has allowed us to study liquid democracy as a form of binary aggregation with abstentions. The second interpretation has led us to study liquid democracy through extremely simple models of opinion diffusion corresponding to the Boolean special case of the stochastic processes of opinion diffusion known as DeGroot processes. We have argued that studying aggregation in liquid democracy through this lens offers important advantages with respect to the handling of delegation cycles and the preservation of individual rationality. Through this second perspective we have also shown how off-the-shelf logical techniques can be used to analyze properties (such as convergence) of the diffusion process underpinning liquid democracy.



\bibliographystyle{plain}
\bibliography{biblio.bib}


\appendix

\section{Binary aggregation (without abstention)} \label{appendix:binary}

The formalism of choice for this paper is binary aggregation \cite{grandi13lifting}.
A binary aggregation structure (\emph{BA structure}) is a tuple $\S = \tuple{\N,\Atoms,\gamma}$ where:
\begin{itemize}
\item $\N = \set{1,\dots,n}$ is a finite set individuals s.t. $|\N|= n \in \mathbb{N}$;
\item $\Atoms = \set{p_1,\dots,p_m}$ is a finite set of issues ($|\Atoms|= m \in \mathbb{N}$), each represented by a propositional atom;
\item $\gamma \in \L$ is an (integrity) constraint, where $\L$ is the propositional language constructed by closing $\Atoms$ under a functionally complete set of Boolean connectives (e.g., $\set{\neg, \wedge}$)
\end{itemize}

An {\em opinion} $O: \Atoms \to \set{\0,\1}$ is an assignment of truth values to the set of issues $\Atoms$, and the set of all opinions is denoted by $\D$. The opinion of an agent $i$ is said to be ``consistent" whenever $O_i \models \gamma$, that is, $i$'s opinion satisfies the integrity constraint. The set of all consistent opinions is denoted $\D_c = \set{O \in \D \mid O \models \gamma}$. 
Thus, $O(p)=\0$ (respectively, \mbox{$O(p)=\1$}) indicates that opinion $O$ rejects (respectively, accepts) the issue $p$. Syntactically, the two opinions correspond to the truth of the literals $p$ or $\neg p$. For $p \in \Atoms$ we write $\pm p$ to denote one element from $\set{p, \neg p}$. An \emph{opinion profile} $\O=(O_1,\dots,O_{n})$ records the opinion, on the given set of issues, of every individual in $\N$. Given a profile $\O$ the $i^{\mathit{th}}$ projection $\O$ is denoted $O_i$ (i.e., the opinion of agent $i$ in profile $\O$).
We also denote by $\O(p)= \set{i \in \N \mid O_{i}(p)= \1}$ the set of agents accepting issue $p$ in profile $\O$ and by $\O(p^-)= \set{i \in \N \mid O_{i}(p)= \0}$. 

Given a BA structure $\S$, an aggregation rule (or {\em aggregator}) for $\S$ is a function $F:(\D_{c})^\N \to \D$, mapping every profile of consistent opinions to one collective opinion in $\D$. 
$F(\O)(p)$ denotes the outcome of the aggregation on issue $p$. A benchmark aggregator is \emph{issue-by-issue strict majority rule} ($\maj$), which accepts an issue if and only if the majority of the population accepts it: 
\begin{align}\label{eq:maj}
\maj(\O)(p)= \1 \IFF |\O(p)| \geq \frac{|\N|+1}{2}.
\end{align}

It is well-known that aggregation by majority does not preserve consistency. The standard example is provided by the discursive dilemma, represented by the BA structure $\tuple{\set{1,2,3},\set{p,q,r},r \lequiv p \land q}$. The profile consisting of $O_1 \models p \land q \land r$, $O_2 \models p \land \neg q \land \neg r$, $O_3 \models \neg p \land q \land \neg r$, returns an inconsistent majority opinion $\maj(\O) \models p \land q \land \neg r$.


\section{Relevant terminology from graph theory} \label{appendix:graph}
 
Let $G = \tuple{\N, R}$ be a graph and $R^{*}$ be the transitive and symmetric closure of $R$.  
A \emph{path} is a sequence of nodes $\tuple{i_1,\dots, i_k}$, such that, for all $l\in\{1,\dots,k\}$, $i_lRi_{l+1}$. 
The \emph{distance} between two nodes $i,j$ is the length of the shortest path $\tuple{i,\dots, j}$ between them. 
The \emph{diameter} of a graph is the maximal distance between any two nodes related by a path.
A \emph{cycle} is a path of length $k$ such that $i_1=i_k$. 
A set of nodes $S\subseteq \N$ is said to be:

\begin{itemize}
\item[]\emph{a cycle in $G$} if all elements in $S$ are in one cycle of length $|S|$, 
\item[]\emph{connected} if for any $i , j \in S$: $i R^{*}j$, 
\item[] \emph{strongly connected} if for any $i,j \in S$: there is a path $\tuple{i,\dots, j}$, 
\item[]\emph{closed} if for any $i\in S$, $j \notin S$, it is not the case that $iRj$, 
\item[] a \emph{connected component} if for any $i,j \in \N$: $iR^{*}j$ if and only if $i,j\in S$, 
\item[] \emph{aperiodic} if the greatest common divisor of the lengths of its cycles is $1$.  
\end{itemize}


\end{document}